\documentclass{article}
\usepackage{graphicx, amssymb, amsmath, amsthm, braket, tikz, color}
\usepackage[colorlinks = true]{hyperref}
\usepackage[normalem]{ulem}
\usepackage[all]{xy}
\CompileMatrices

\DeclareMathOperator{\tr}{trace}
\DeclareMathOperator{\diag}{diag}
\DeclareMathOperator{\nbd}{nbd}

\DeclareMathOperator{\supp}{supp}

\newtheorem{theorem}{Theorem}
\newtheorem{corollary}{Corollary}
\newtheorem{lemma}{Lemma}
\newtheorem{example}{Example}
\newtheorem{definition}{Definition}

\title{Condition for zero and non-zero discord in graph Laplacian quantum states}
\author{
	Supriyo Dutta \\ Department of Mathematics\\ Indian Institute of Technology Jodhpur.\\ Email: \texttt{dosupriyo@gmail.com}\vspace{.25cm}\\
	Bibhas Adhikari\\ Department of Mathematics\\ Indian Institute of Technology Kharagpur.\\Email: \texttt{bibhas@maths.iitkgp.ernet.in} \vspace{.25cm}\\
	Subhashish Banerjee\\ Department of Physics\\ Indian Institute of Technology Jodhpur.\\ Email: \texttt{subhashish@iitj.ac.in}
}
\date{}

\begin{document}

	\maketitle
	
	\begin{abstract}
		This work is at the interface of graph theory and quantum mechanics. Quantum correlations epitomize the usefulness of quantum mechanics. Quantum discord is an interesting facet of bipartite quantum correlations. Earlier, it was shown that every combinatorial graph corresponds to quantum states whose characteristics are reflected in the structure of the underlined graph. A number of combinatorial relations between quantum discord and simple graphs were studied. To extend the scope of these studies, we need to generalise the earlier concepts applicable to simple graphs to weighted graphs, corresponding to a diverse class of quantum states. To this effect, we determine the class of quantum states whose density matrix representation can be derived from graph Laplacian matrices associated with a weighted directed graph and call them graph Laplacian quantum states. We find the graph theoretic conditions for zero and non-zero  quantum discord for these states. We apply these results on some important pure two qubit states, as well as a number of mixed quantum states, such as the Werner, Isotropic, and $X$-states. We also consider graph Laplacian states corresponding to simple graphs as a special case.\vspace{.25cm}
		
	\end{abstract}

	\section{Introduction}
	
		A strong interest has been generated towards applying graphs and networks in different aspects of quantum mechanics and quantum information in recent times \cite{braunstein2006laplacian, hassan2007combinatorial, dutta2016graph, lockhart2016combinatorial, belhaj2016weighted, simmons2017symmetric, simmons2017quantum, belhaj2016multi, wu2010graphs, dutta2017quantum}. Another scenario where this could be envisaged would be the discretization of a Hamiltonian, within a tight binding model, such that the electron lives on the vertices of the graph and the dynamics induces their hopping from one vertex to another. Quantum states defined by using combinatorial Laplacian and signless Laplacian matrices associated with a weighted directed graph have recently been introduced and investigated in \cite{adhikari2017Laplacian}. Throughout the paper we call these quantum states as graph Laplacian states. In this paper we derive certain conditions on structure of weighted directed graphs such that the graph Laplacian quantum states have zero quantum discord.
		
		Recall that quantum discord $D(\rho)$ of a state $\rho$ is a class of quantum correlations which has been used as a resource in quantum information and communication \cite{zurek2000einselection, ollivier2001quantum, henderson2001classical, adhikari2012operational, pirandola2013quantum, brodutch2016should}. From the perspective of computational complexity, it is proved that calculating $D(\rho)$ is an NP-complete problem \cite{huang2014computing, lim2014sudden}. This calls for developing alternate measures and techniques to realize quantum discord \cite{dakic2010necessary}. In our earlier work \cite{dutta2017quantum}, we have constructed a number of criterion for zero discord in graph Laplacian quantum states corresponding to simple graph. Also, we have produced a graph theoretic measure of quantum discord. Simple graphs do not possess directed, weighted edges and loops. As a result we can express a limited number of useful quantum states as a graph Laplacian quantum state. Therefore, for wider applicability, we need to generalise our earlier constructions for weighted graphs. This work is related to combinatorial graphs and their corresponding quantum states. We find out conditions on graphs such that the corresponding quantum states have non-zero discord. These conditions shed light into the nature of discord in a number of important quantum states. Hence just by observing the structural properties of the graph, the quantum discord can be determined. These properties include existence or non-existence of some particular edges, and degree of vertices. Therefore, this work develops a new method of visualization to the problem of discord by exploiting the connection between graph theory and quantum mechanics.
		
		Consider a bipartite system of order $m\times n.$ Then the density matrix corresponding to such a bipartite system is of order $mn\times mn$ and it is a block matrix having each block of size $n$. Treating a graph as a clustered graph on $m\times n$ vertices in which each cluster contains $n$ vertices, the combinatorial Laplacian matrix and the signless Laplacian matrix (defined in Section 2)  define block density matrices corresponding to the graph. The arrangement of edges in the clustered graph determines the discord in the corresponding quantum states. This approach can be considered as a combinatorial approach to the realization of quantum discord. Interpreting the Werner states, isotropic states  and some of the $X$ states as graph Laplacian states, we illuminate their quantum discord in graph theoretic terms.
		
		The article is organized as follows. We provide a brief overview of graph theory which is required for the remaining part of the article. We establish the condition for a quantum state to be a graph Laplacian state. In Section 3, we derive a number of results to generate the graph theoretic criterion of discord. Finally we employ these results on some well known states, both pure and mixed, for example, two qubit graph Laplacian states, Werner, Isotropic, and $X$ states, as well as graph Laplacian states corresponding to arbitrary simple graphs. We then conclude.

	\section{Preliminaries}
	
		A weighted digraph is an ordered pair of sets denoted by $G = (V(G), E(G))$ where $V(G)$ is called the vertex set and $E(G)\subseteq V\times V$ is a set of ordered pair of vertices called the edge set with a weight function $w_G : E(G) \rightarrow \mathbb{C}\setminus \{0\}$ \cite{adhikari2017Laplacian}. If there is no confusion regarding the underlined digraph $G$, we simply denote $w$  for $w_G$. Now, we consider the following assumptions on all weighted digraphs considered in the article.\\		
		{\bf  Assumptions:}
		\begin{enumerate}
			\item 
			Given two vertices $i$ and $j$, if $(i,j)\in E(G)$ then $(j,i)\in E(G)$ and $w(j,i) = \overline{w(i,j)}$,  the complex conjugate of $w(i,j)$.
			\item
			If $(i,i) \in E(G)$ then $w(i,i)\in \mathbb{R}^+$, the set of non-negative real numbers.
		\end{enumerate}
		Note that if $w(i,j) = 1$ for all $(i,j) \in E(G)$ and $(i,i) \notin E(G)$ for all $i\in V(G)$ then the digraph $G$ becomes a simple graph. Let $G$ be a weighted digraph on $N$ vertices. Then the adjacency matrix $A(G) = [a_{ij}]_{N \times N}$ associated with a weighted digraph $G$ is defined by,
		\begin{equation}
		a_{ij} = \begin{cases} w(i,j) & ~\text{if}~ (i,j) \in E(G), \\ 
		0 & ~\text{if}~ (i,j) \notin E(G) \end{cases}
		\end{equation} and $a_{ji}=\overline{a_{ij}}=\overline{w(i,j)}.$ Thus $A(G)$ is a Hermitian matrix. The weighted degree of a vertex $i$ is defined by
		\begin{equation}
		d_{i} = \sum_{(i,j) \in E(G)} |w(i,j)| = \sum_{j = 1}^N |a_{i,j}|.
		\end{equation}
		The degree matrix is the diagonal matrix defined by $D(G) = \diag\{d_1, d_2, \dots d_N\}$. The Laplacian and the signless Laplacian matrices are defined by $L(G) = D(G) - A(G)$ and $Q(G) = D(G) + A(G)$, respectively. It is proved in \cite{adhikari2017Laplacian} that $L(G)$ and $Q(G)$ are positive semidefinite Hermitian matrices. Recall that a density matrix $\rho$ corresponding to a quantum state is a positive semi-definite Hermitian matrix with unit trace. Consequently the density matrices corresponding to a weighted digraph are defined as
		\begin{equation}
		\rho_l(G) = \frac{1}{\tr(L(G))} L(G) ~\text{and}~ \rho_q(G) = \frac{1}{\tr(Q(G))} Q(G).
		\end{equation}
		We denote $\rho_l(G)$ and $\rho_q(G)$ together with $\rho(G)$ when no confusion arises. We call the digraph $G$ as the graph representation of $\rho(G).$ Then we have the following lemma.
		
		\begin{lemma}\label{graphical_states}
			A density matrix $\rho = (\rho_{ij})_{N \times N}$ has a graph representation if and only if for all $i$ and $j$, $\rho_{ii} \ge \sum_{i \neq j} |\rho_{i,j}|$.
		\end{lemma}
		
		\begin{proof}
			If $\rho$ has a graph representation, the weighted digraph $G$ has $N$ vertices since the order of $\rho$ is $N$. When $i \neq j$ and $\rho_{ij} \neq 0$ there is a directed edge $(i,j)$ with edge weight $w(i,j) = \rho_{ij}$. As $\rho$ is a positive semi-definite Hermitian matrix, $\rho_{ji} = \overline{\rho_{ij}}$ and $\rho_{ii}$ is a non-negative real number. Thus $(i,j)$ and $(j,i)$ exists together with $w(j,i) = \overline{w(i,j)}$.  Besides, $\rho_{ii} = d_i + s a_{ii}$. Here $s = -1$ for $\rho_l(G)$ and $s = 1$ for $\rho_q(G)$. Now,
			\begin{equation}
				\begin{split} 
					\rho_{ii} = d_i + s a_{ii} & = \sum_{j = 1}^N |w(i, j)| + s w(i,i) = \sum_{j \neq i} |w(i,j)| + |w(i,i)| + s w(i,i)\\
					& = \sum_{j \neq i} |\rho_{ij}| + |w(i,i)| + s w(i,i).
				\end{split} 
			\end{equation}
			As $\rho_{ii}$ is real, $w(i,i)$ must be real in the above expression. Now two cases arise. 
			\begin{itemize}
				\item 
				Case-I: Let $w(i,i) = 0$ or $|w(i,i)| = -s w(i,i)$. In any case, $\rho_{ii} = \sum_{j \neq i} |\rho_{i,j}|$.
				\item
				Case-II: Let $w(i,i) \neq 0$ and $|w(i,i)| + s w(i,i) = 2|w(i,i)|$. Then, from the above equation,
				\begin{equation}\label{graphical_2}
				|w(i,i)| = \frac{\rho_{ii} - \sum_{i \neq j} |\rho_{ij}|}{2} \ge 0. 
				\end{equation}
				In this case, $\rho_{ii} \ge \sum_{i \neq j} |\rho_{ij}|$.
			\end{itemize}
		\end{proof}
		
		Matrices with the property $\rho_{ii} \ge \sum_{i \neq j} |\rho_{ij}|$ are called diagonally dominant. Entanglement of the diagonally dominant density matrices were studied in \cite{wu2006conditions}. Now we state the following definition and examples.
		\begin{definition}\label{graphical_state_definition}
			{\bf Graph Laplacian quantum states:} A quantum state $\rho$ is said to be a graph Laplacian state  if there exists a weighted directed graph $G$ such that $\rho=\rho(G).$ 
		\end{definition}
		Note that any graph Laplacian state satisfies Lemma \ref{graphical_states} that is every graph Laplacian quantum state is represented by a diagonally dominant density matrix and the converse also holds.
		
		\begin{example}
			Consider the quantum state
			$$\rho = \frac{1}{5} \ket{0}\bra{0} + \frac{2}{5} \ket{0}\bra{1} + \frac{2}{5} \ket{1}\bra{0} + \frac{4}{5} \ket{1}\bra{1} = \frac{1}{5} \begin{bmatrix} 1 & 2 \\ 2 & 4 \end{bmatrix}.$$
			Note that, $\rho_{11} \le \rho_{12}$. Thus, this quantum state is not a graph Laplacian quantum state.
		\end{example}
		
		From now onwards we use the word graph for weighted directed graph. The following describes the framework for interpreting any graph $G$ on $mn$ vertices as a graph with $m$ clusters each with $n$ vertices. Recall that cluster of a graph is a subgraph of the graph. Consider a partition of the vertex set $V(G)$ that produces these clusters as follows.
		\begin{equation}\label{clustering}
		\begin{split}
		& V = C_1 \cup C_2 \cup \dots \cup C_m;\\
		& C_\mu \cap C_\nu = \emptyset ~\text{for}~ \mu \neq \nu ~\text{and}~ \mu, \nu = 1, 2, \dots m;\\
		& C_\mu = \{v_{\mu 1}, v_{\mu 2}, \dots v_{\mu n}\}.
		\end{split}
		\end{equation}
		For any vertex $v_{\gamma i}$, the Roman index $i$ represents the position of a vertex in $\gamma$-th cluster indexed by a Greek index. A density matrix $\rho$ acting on $\mathcal{H}^{(m)} \otimes \mathcal{H}^{(n)}$ has order $m \times n$, which corresponds to a graph with $m \times n$ vertices. Clearly dimension of $\mathcal{H}^{(n)}$ is the number of vertices in the cluster $C_\mu$. Also the number of clusters in the graph is the dimension of $\mathcal{H}^{(m)}$. Then, 
		\begin{equation}\label{adjacency}
		A(G) = \begin{bmatrix} A_{11} & A_{12} & \dots & A_{1m} \\ A_{21} & A_{22} & \dots & A_{2m} \\ \vdots & \vdots & \ddots & \vdots \\ A_{m1} & A_{m2} & \dots & A_{mm}\end{bmatrix}_{m \times m},
		\end{equation} 
		where $A_{\mu \nu}$ are blocks of order $n \times n$ \cite{dutta2016bipartite}. Note that, $A_{\mu \mu}$ contains the weights of the edges joining two vertices inside the cluster $C_\mu$. Also $A_{\mu\nu}$ contains the weights of the edges joining vertices of $C_\mu$ and $C_\nu$. As $A(G)$ is a Hermitian matrix, we have $A_{\mu \mu}^\dagger = A_{\mu \mu}$ and $A_{\mu \nu}^\dagger = A_{\nu \mu}$, for $\mu \neq \nu$. Now we recall the definition of induced subgraph \cite{west2001introduction}.
		
		\begin{definition}
			{\bf Induced Subdigraph:} A subdigraph $H$ of a digraph $G$ is an induced subdigraph if $u, v \in V(H)$ and $(u, v) \in E(G)$ implies $(u, v) \in E(H)$.
		\end{definition}
		
		We denote the induced subdigraph generated by the cluster $C_\mu$ by $\langle C_\mu \rangle$. Also the subdigraph $\langle C_\mu, C_\nu \rangle$ consists of all the vertices in $C_\mu \cup C_\nu$, and all the edges $(u, v)$ and $(v,u)$ with $u \in C_\mu$, and $v \in C_\nu$. In the equation (\ref{adjacency}), the block $A_{\mu \mu}$ acts as the adjacency matrix of the cluster $C_\mu$. Further, $A_{\mu \nu}$ represents all the edges joining two vertices belonging to different clusters $C_\mu$ and $C_\nu$. Thus, 
		\begin{equation}
		A(\langle C_\mu, C_\nu \rangle) = \begin{bmatrix} 0 & A_{\mu \nu} \\ A_{\nu \mu} & 0\end{bmatrix} = \begin{bmatrix} 0 & A_{\mu \nu} \\ A_{\mu \nu}^\dagger & 0\end{bmatrix}.
		\end{equation}
		This clustering on $V(G)$ also partitions the degree matrix into blocks, such that, $D = \diag\{D_1, D_2, \dots D_m\}$, where $D_\mu$ is a diagonal matrix containing degree of the vertices in $C_\mu$. If $B_{\mu \nu}$ are blocks of the density matrix $\rho(G)$, then
		\begin{equation}\label{block_and_graph}
		B_{\mu \nu} = \begin{cases}
		s\frac{A_{\mu \nu}}{d} & ~\text{if}~ \mu \neq \nu\\
		\frac{D_\mu + s A_{\mu \mu}}{d} & ~\text{if}~ \mu = \nu
		\end{cases},
		\end{equation}
		where $s = 1$ and $-1$ for $\rho(G) = \rho_q(G)$ and $\rho_l(G)$, respectively.
		
		A bipartite density matrix acts on the Hilbert space $\mathcal{H}^{(A)} \otimes \mathcal{H}^{(B)}$ where $\mathcal{H}^{(A)}$ and $\mathcal{H}^{(B)}$ denote the state spaces (Hilbert spaces) corresponding to the constituent systems $A$ and $B$ respectively,  $\otimes$ denotes the Kronecker (tensor) product. For any bipartite density matrix $\rho$, there are two reduced density matrices $\rho_a$ and $\rho_b$ acting on the spaces $\mathcal{H}^{(A)}$ and $\mathcal{H}^{(B)}$, respectively.  The mutual information in $\rho$ is defined as $I(\rho) = S(\rho_a) + S(\rho_b) - S(\rho)$ where $S(X) = -\tr(X\log(X))$ is the von-Neumann entropy of a state $X$. Let the set of all possible von Neumann measurements with respect to the system $\mathcal{H}^{(B)}$ be $\Pi^b$. Thus we obtain an another measure of mutual information in $\rho$ given by $I(\rho|\Pi^b) = S(\rho_a) - S(\rho|\Pi^b)$, where $S(\rho|\Pi^b) = \sum_kp_kS(\rho_k)$, and $\rho_k = \frac{1}{p_k}(I_a \otimes \Pi_k^b) \rho (I_a \otimes \Pi_k^b)$ with $p_k = \tr[(I_a \otimes \Pi_k^b) \rho (I_a \otimes \Pi_k^b)]$, $k = 1, 2, \dots \dim(\mathcal{H}^b)$. Quantum discord is the difference between these two classically equivalent measures $I(\rho)$ and $I(\rho|\Pi^b)$ \cite{henderson2001classical, ollivier2001quantum}. There are quantum states which yield equal value for both the measures  that are known as classical-quantum states \cite{kus2009classical} or pointer states.
		
		\begin{definition}{\bf Quantum discord:} 
			The quantum discord of a bipartite state $\rho$ is
			\begin{equation}
			D(\rho) = \min_{\Pi^b}\{I(\rho) - I(\rho|\Pi^b)\}.
			\end{equation}
		\end{definition}

	\section{Graph theoretic criterion for quantum discord}
	
		Before going to the graph theoretic aspects of quantum discord, below we discuss a number of definitions and lemmas which will be used for further derivations. Given a vertex $i$ we call the set $\nbd_G(i) = \{j: j \in V(G), (i,j) \in E(G)\}$ as the neighborhood of vertex $i$. Under the basic assumptions outlined above, $(i,j)$ and $(j,i)$ belong to $E(G)$ together. With respect to the vertex $i$ we describe $(i,j)$ as the outgoing edge and $(j,i)$ as the incoming edge. We collect the weights of the edges incident to vertex $i$ in the following sets.
		\begin{equation}
		\begin{split} 
		W(\nbd_G(i)_{out}) & = \{w_G(i,j):  (i,j) \in E(G)\},\\
		W(\nbd_G(i)_{in}) & = \{w_G(j,i): (j,i) \in E(G)\}.
		\end{split}
		\end{equation} 
		
		\begin{definition}\label{support_of_a_vector}
			{\bf Support of a vector:} Given a vector $a \in \mathbb{C}^n$ there is a set $\supp(a)$ defined by,
			$$\supp(a) = \{i: a(i) \neq 0\}$$ where $a(i)$ denotes the $i$th entry of $a.$
		\end{definition} 
		Given two vectors $a, b \in \mathbb{C}^n$ we define their product as,
		\begin{equation}\label{inner_product}
		\braket{a|b} = \sum_{ k \in \supp(a) \cap \supp(b)} a(k) b(k).
		\end{equation}
		This should not be confused with the inner product between two state vectors. Given a matrix $A = (a_{ij})_{n \times n}$, $a_{i*}$ and $a_{*j}$ denotes the $i$-th row and $j$-th column vectors, respectively. Corresponding to every $A$, there is a weighted bipartite graph of order $2n$, $\mathcal{A} = (V(\mathcal{A}), E(\mathcal{A}))$ with the adjacency matrix,
		\begin{equation}\label{matrix_to_bipartite}
		A(\mathcal{A}) = \begin{bmatrix} 0 & A \\ A^\dagger & 0 \end{bmatrix}.
		\end{equation}
		As $\mathcal{A}$ is a bipartite graph we can write $V(\mathcal{A}) = C_\mu \cup C_\nu$, where $C_\mu = \{v_{\mu 1}, v_{\mu 2}, \dots v_{\mu n}\}$,	$C_\nu = \{v_{\nu 1}, v_{\nu 2}, \dots v_{\nu n}\}$ and $C_\mu \cap C_\nu = \emptyset$ as mentioned in equation (\ref{clustering}). Therefore, $\mathcal{A} = \langle C_{\mu}, C_\nu \rangle$, the subgraph generated by the vertex sets $C_\mu$ and $C_\nu$. The directed edge $(v_{\mu i}, v_{\nu j}) \in E(\mathcal{A})$, if and only if $a_{ij} \neq 0$. Also, $w(v_{\mu i}, v_{\nu j}) = a_{ij}$. Moreover, the adjacency matrix $A(\mathcal{A})$ indicates the existence of $(v_{\nu j}, v_{\mu i})$ with $w(v_{\nu j}, v_{\mu i}) = \overline{a_{ij}}$. Now, 
		\begin{equation}
		\nbd_{\mathcal{A}}(v_{\mu i}) = \{v_{\nu j}: (v_{\mu i}, v_{\nu j}) \in E(G)\}  \subset C_\nu.
		\end{equation}  
		Similarly, $\nbd_{\mathcal{A}}(v_{\nu i}) \subset C_\mu$. Let $0_{1,n}$ and $0_{n,1}$ are zero row and column vectors. Note that, the $i$-th row of $A(\mathcal{A})$, that is $(0_{1,n}, a_{i*})$ represents weights of outgoing edges from the vertex $v_{\mu i}$. According to the definition \ref{support_of_a_vector}, $\supp(0_{1,n}, a_{i*}) = \supp(a_{i*})$ which represents indexes of vertices in $\nbd_{\mathcal{A}}(v_{\mu i})$. Thus we have,
		\begin{equation}
		\supp(a_{i*}) = \nbd_{\mathcal{A}}(v_{\mu i}), ~\text{and}~ a_{i*} = W(\nbd_{\mathcal{A}}(v_{\mu i})_{out}).
		\end{equation}
		Similarly, the $(n + i)$-th column of $A(\mathcal{A})$, that is $(a_{*i}, 0_{n,1})$ represents edge weights of the incoming edges to the vertex $v_{\nu i}$. Also, $\supp(a_{*i})$ represents indexes of vertices in $\nbd_{\mathcal{A}}(v_{\nu i})$. Hence,
		\begin{equation}
		\supp(a_{*i}) = \nbd_{\mathcal{A}}(v_{\nu i}) ~\text{and}~ a_{*i} = W(\nbd_{\mathcal{A}}(v_{\nu i})_{in}).
		\end{equation}
		
		In particular, any complex Hermitian matrix $A$  of order $n$ can be considered as an adjacency matrix of a graph \~{A}, where $V(\text{\~{A}}) = C_\mu = \{v_{\mu,1}, v_{\mu,2}, \dots v_{\mu,n}\}$. The edge $(v_{\mu,i}, v_{\mu,j}) \in E(\text{\~{A}})$ if and only if $a_{ij} \neq 0$. Thus, \~{A} $ = \langle C_\mu \rangle$, the induced subgraph generated by the vertex set $C_\mu$. Here, the row vector $a_{i*}$ represents all outgoing edges from the vertex $v_{\mu i}$. Thus, $\supp(a_{i*}) = \nbd_{\text{\~{A}}}(v_{\mu i})$. Similarly, $\supp(a_{*i}) = \nbd_{\text{\~{A}}}(v_{i \mu})$. Now, we have the following results. Interested readers may go through their proofs in the Appendix.
		
		\begin{lemma}\label{commutativity}
			Let the weighted bipartite digraphs corresponding to complex square matrices $A$ and $B$ of order $n$ be $\mathcal{A} = \langle C_\mu, C_\nu \rangle$, and $\mathcal{B} = \langle C_\alpha, C_\beta \rangle$, respectively. The matrices $A$ and $B$ commute, if and only if for all $i, j$ with $1 \le i,j \le n$,
			$$\sum_{k \in \nbd(v_{\mu i}) \cap \nbd(v_{\beta j})} w(v_{\mu i}, v_{\nu k}) w(v_{\alpha k}, v_{\beta j}) = \sum_{k \in \nbd(v_{\alpha i}) \cap \nbd(v_{\nu j})} w(v_{\alpha i}, v_{\beta k}) w(v_{\mu k}, v_{\nu j}).$$
		\end{lemma}
	
		Note that, if $\langle C_\mu, C_\nu \rangle = \langle C_\alpha, C_\beta \rangle$ then the condition of commutativity holds. Here, equality between two graphs indicates that they have equal vertex sets, edge sets, and vertex labellings. Also, if any of the graphs be empty, then the commutativity condition holds trivially.
		
		\begin{corollary}\label{commutativity_1}
			Let \~{A} $= \langle C_\mu \rangle$, and $\mathcal{B} = \langle C_\alpha, C_\beta \rangle$ be graphs corresponding to a Hermitian matrix $A = (a_{ij})_{n \times n}$, and square matrix $B = (b_{ij})_{n \times n}$. They commute if and only if for all $i, j$ with $1 \le i,j \le n$,
			$$\sum_{k \in \nbd(v_{\mu i}) \cap \nbd(v_{\beta j})} w(v_{\mu i}, v_{\mu k}) w(v_{\alpha k}, v_{\beta j}) = \sum_{k \in \nbd(v_{\alpha i}) \cap \nbd(v_{\mu j})} w(v_{\alpha i}, v_{\beta k}) w(v_{\mu k}, v_{\mu j}).$$
		\end{corollary}

		\begin{corollary}\label{commutativity_2}
			Two Hermitian matrices $A = (a_{ij})_{n \times n}$, and $B = (b_{ij})_{n \times n}$ corresponding to graphs \~{A} $= \langle C_\mu \rangle$, and \~{B} $ = \langle C_\nu \rangle$ commute, if and only if for every $i, j$ with $1 \le i, j \le n$,
			$$\sum_{k \in \nbd(v_{\mu i}) \cap \nbd(v_{\nu j})} w(v_{\mu i}, v_{\mu k}) w(v_{\nu k}, v_{\nu j}) = \sum_{k \in \nbd(v_{\nu i}) \cap \nbd(v_{\mu j})} w(v_{\nu i}, v_{\nu k}) w(v_{\mu k}, v_{\mu j}).$$
		\end{corollary}
		
		A complex normal matrix $A$ commutes with its conjugate transpose, that is $AA^\dagger = A^\dagger A$. Hermitian matrices are trivially normal matrices. But there are normal matrices which are not Hermitian.
		\begin{lemma}\label{normality}
			Let $\mathcal{A} = \langle C_\mu, C_{\nu} \rangle$ be a weighted bipartite digraph corresponding to a matrix $A = (a_{ij})_{n \times n}$. It is normal, if and only if for every $i$, and $j$ with $1 \le i, j \le n$,
			$$\sum_{k \in \nbd(v_{\mu i}) \cap \nbd(v_{\mu j})} w(v_{\mu i}, v_{\nu k}) w(v_{\nu k}, v_{\mu j}) = \sum_{k \in \nbd(v_{\nu i}) \cap \nbd(v_{\nu j})} w(v_{\nu i}, v_{\mu k}) w(v_{\mu k}, v_{\nu j}).$$
		\end{lemma}	
		
		Now we consider a trivial observation related to the above lemma, which will be used later. Let there be only one edge of arbitrary non-zero weight, $(v_{\mu, p}, v_{\nu, q})$ with $p \neq q$, between two clusters $C_\mu$ and $C_\nu$. Now, for $i = j = p$,
		\begin{equation}
		\sum_{k \in \nbd(v_{\mu i}) \cap \nbd(v_{\mu j})} w(v_{\mu i}, v_{\nu k}) w(v_{\nu k}, v_{\mu j}) = w(v_{\mu p}, v_{\nu q}) w(v_{\nu q}, v_{\mu p}).
		\end{equation} 
		Also, for $i = j = p$ the set $\nbd(v_{\nu i}) \cap \nbd(v_{\nu j}) = \emptyset$, as $v_{\nu p}$ is an isolated vertex. Hence, the term $\sum_{k \in \nbd(v_{\nu i}) \cap \nbd(v_{\nu j})}$ $w(v_{\nu i}, v_{\mu k}) w(v_{\mu k}, v_{\nu j})$ takes no value. In this case, the graph $\langle C_\mu, C_\nu \rangle$ fails to fulfil the normality condition. Note that, for $p = q$ the graph $\langle C_\mu, C_\nu \rangle$ with single edge $(v_{\mu, p}, v_{\nu, q})$ represents a normal matrix.
		
		Let us recollect some important facts discussed above. The Lemma \ref{graphical_states} provides conditions on a quantum state $\rho$ acting on $\mathcal{H}^{(m)} \otimes \mathcal{H}^{(n)}$ to have a graph representation. In the equation (\ref{clustering}), we partition a digraph with $N = mn$ vertices into clusters. Thus two sets of subgraphs are generated: $\{\langle C_{\mu} \rangle: \mu = 1, 2, \dots m\}$ where every graph is of order $n$, and $\{\langle C_{\mu}, C_{\nu} \rangle: \mu, \nu = 1, 2, \dots m; \mu \neq \nu\}$ where every graph is bipartite graph with $2n$ vertices.

		\begin{theorem}\label{math_thm}
			The quantum state corresponding to the density matrix $\rho(G)$ has zero discord if and only if	\begin{enumerate}
				\item 
				\textbf{Commutativity condition:} Given any two subgraphs $\langle C_\mu, C_\nu \rangle$, and $\langle C_\alpha, C_\beta \rangle$ and for all $i, j$ with $1 \le i,j \le n$,
				$$\sum_{k \in \nbd(v_{\mu i}) \cap \nbd(v_{\beta j})} w(v_{\mu i}, v_{\nu k}) w(v_{\alpha k}, v_{\beta j}) = \sum_{k \in \nbd(v_{\alpha i}) \cap \nbd(v_{\nu j})} w(v_{\alpha i}, v_{\beta k}) w(v_{\mu k}, v_{\nu j}).$$
				\item
				\textbf{Normality condition:} For all subgraph $\langle C_\mu, C_{\nu} \rangle$ and for every $i$, and $j$ with $1 \le i, j \le n$,
				$$\sum_{k \in \nbd(v_{\mu i}) \cap \nbd(v_{\mu j})} w(v_{\mu i}, v_{\nu k}) w(v_{\nu k} v_{\mu j}) = \sum_{k \in \nbd(v_{\nu i}) \cap \nbd(v_{\nu j})} w(v_{\nu i}, v_{\mu k}) w(v_{\mu k} v_{\nu j}).$$
				\item
				\textbf{Degree condition} The graph satisfies the following two degree criterion,
				\begin{enumerate}
					\item
					\begin{equation}
						\begin{split} 
							\pm & \big[w(v_{\nu i},v_{\nu j})(d_{\mu i} - d_{\mu j}) + w(v_{\mu i}, v_{\mu j})(d_{\nu j} - d_{\nu i}) \big] \\
							+ & \sum_{k \in \nbd(v_{\mu i}) \cap \nbd(v_{\nu j})} w(v_{\mu i}, v_{\mu k}) w(v_{\nu k}, v_{\nu j})	\\
							- &\sum_{k \in \nbd(v_{\nu i}) \cap \nbd(v_{\mu j})} w(v_{\nu i}, v_{\nu k}) w(v_{\mu k}, v_{\mu j}) = 0,
						\end{split} 
					\end{equation}
					\item 
					\begin{equation}
						\begin{split} 
							& w(v_{\alpha i}, v_{\beta j}) (d_{\mu i} - d_{\mu j}) \pm \\ 
							& \big[ \sum_{k \in \nbd(v_{\mu i}) \cap \nbd(v_{\beta j})} w(v_{\mu i}, v_{\mu k}) w(v_{\alpha k}, v_{\beta j}) \\
							& - \sum_{k \in \nbd(v_{\alpha i}) \cap \nbd(v_{\mu j})} w(v_{\alpha i}, v_{\beta k}) w(v_{\mu k}, v_{\mu j}) \big] = 0.
						\end{split} 
					\end{equation}
				\end{enumerate}
			\end{enumerate}
		\end{theorem}

	\section{Discord of some graph Laplacian states}
	
		In this section, we illustrate the discussions in the previous sections by studying discord of some graph Laplacian states corresponding to well known quantum states. Before that, we make a list of graph properties that are useful for determining the quantum discord.
		
		The lemma \ref{commutativity} and \ref{normality} as well as their corollaries suggest arrangements of edge weights in the graph, necessary for zero discord. In addition, the degree of the vertices should fulfil the degree condition of the theorem \ref{math_thm}. If vertices of $C_\mu, \mu = 1, 2, \dots m$ have equal degree, then $d_{\mu i} - d_{\mu j} = 0$ independent of the existence of edge $(v_{\alpha i}, v_{\beta j})$. In this case, every pair of subgraph $\langle C_\mu \rangle$ and $\langle C_\alpha, C_\beta \rangle$ satisfies the Corollary \ref{commutativity_1}.
		
		Also if the subgraphs $\langle C_\mu \rangle$ and $\langle C_\nu \rangle$ fulfill the commutativity condition described in the corollary \ref{commutativity_2} then the first degree condition takes the following simpler form:
		\begin{equation}
		w(v_{\nu i},v_{\nu j})(d_{\mu i} - d_{\mu j}) + w(v_{\mu i}, v_{\mu j})(d_{\nu j} - d_{\nu i})  = 0 ~\text{for all}~ i, j.
		\end{equation} 
		Further, if the subgraphs $\langle C_\alpha, C_\beta \rangle$ and $\langle C_\mu \rangle$ satisfy the commutativity condition described in the corollary \ref{commutativity_1}, the equation is simplified to 
		\begin{equation}
		w(v_{\alpha i}, v_{\beta j}) (d_{\mu i} - d_{\mu j}) = 0 ~\text{for all}~ i, j.
		\end{equation}
		When the graph is a simple graph, $w(u,v) = 1$ for all $(u,v) \in E(G)$. Then this condition on weighted graphs is consistent with that of simple graphs \cite{dutta2017quantum}.

		\subsection{Two qubit pure states}
	
			Two qubit quantum states are the simplest bipartite quantum states. Here, we consider two examples of 2-qubit pure states: $\ket{\psi_1} = a\ket{00} + b\ket{11}$, and $\ket{\psi_2} = a\ket{00} + b\ket{01}$, where $|a|^2 + |b|^2 = 1$. Restricting $a$ and $b$ in $\ket{\psi_1}$ to $\frac{1}{\sqrt{2}}$ leads to the well known Bell state. The density matrices corresponding to these quantum states are 
			\begin{equation}
			\sigma_1 = \ket{\psi_1}\bra{\psi_1} = \begin{bmatrix} a^2 & 0 & 0 & ab \\ 0 & 0 & 0 & 0 \\ 0 & 0 & 0 & 0 \\ ab & 0 & 0 & b^2 \end{bmatrix}, ~\text{and}~  \sigma_2 = \ket{\psi_2}\bra{\psi_2} = \begin{bmatrix} a^2 & ab & 0 & 0 \\ ab & b^2 & 0 & 0  \\ 0 & 0 & 0 & 0 \\ 0 & 0 & 0 & 0 \end{bmatrix},
			\end{equation}
			respectively.
			
			From Lemma 1, we can see that these density matrices represent graph Laplacian quantum states if $a^2 \geq ab$ and $b^2 \geq ab$. If $a \neq 0$ and $b \neq 0$ then these two inequalities together imply $a = b$. A density matrix of order $4$ corresponds to a graph with four vertices. Also, the graphs representing 2-qubit bipartite states must have two clusters. Let $\sigma_1 = \rho_q(G_1)$, and $\sigma_2 = \rho_q(G_2)$. Then,
			\begin{equation}
			G_1 \equiv \xymatrix{\bullet_{00} \ar@{-}[dr] & \bullet_{01} \\ \bullet_{10} & \bullet_{11}} ~\text{and}~ G_2 = \xymatrix{\bullet_{00} \ar@{-}[r] & \bullet_{01} \\ \bullet_{10} & \bullet_{11}}
			\end{equation}
			
			From the conditions in Theorem 1, we can visualize that the graph $G_1$ violates the normality condition. But, the graph $G_2$ satisfies all of them. Hence, the state $\sigma_2$ has zero discord, but $\sigma_1$ has non-zero discord. This example clearly indicates that quantum discord of states depend on the distribution of edges in the graph.
		
		\subsection{Werner state}
			The Werner state is a well-known bipartite mixed quantum state. A Werner state \cite{werner1989quantum} is represented by,
			\begin{equation}\label{werner}
			\rho_{x,d} = \frac{d - x}{d^3 - d}I + \frac{xd - 1}{d^3 - d}F,
			\end{equation}
			where $F = \sum_{i,j}^d \ket{i}\bra{j} \otimes \ket{j}\bra{i}$, $x \in [0, 1]$ and $d$ is the dimension of the individual subsystems. Note that, $\rho_{x,d}$ is a real symmetric matrix of order $d^2$. These are separable Werner states with non-zero quantum discord \cite{luo2008using}. 
			
			We show that all the Werner states are graph Laplacian states. As $\rho_{x,d}$ acts on the space $\mathcal{H}^{(d)} \otimes \mathcal{H}^{(d)}$, we partition the vertex set into $d$ clusters $C_\mu, \mu = 1, 2, \dots d$, each having $d$ vertices. The corresponding digraph has three types of edges:
			\begin{enumerate}
				\item 
				Loops at diagonal vertices $v_{11}, v_{22}, \dots v_{dd}$ having loop weights $w(v_{\mu, \mu}, v_{\mu, \mu}) = (d -1) + (d - 1)x$.
				\item
				Loops at non-diagonal vertices $\{v_{\mu , i}: \mu \neq i\}$ having loop weights $w(v_{\mu , i}, v_{\mu , i}) = d - x$.
				\item
				Non-loop edges with weight $w(v_{\mu,i}, v_{i, \mu}) = dx - 1$. Note that, there is only one edge between two different clusters. All such edges are diagonal and parallel to each-other.
			\end{enumerate}
			The following example would help to illustrate this structure.
			
			\begin{example}
				We may represent $\rho_{x,3}$, and $\rho_{x,4}$ as a graph with $9$ and $16$ vertices depicted in figure \ref{werner_graph_1} and \ref{werner_graph_2}. The edge weights $a, b$, and $c$ represents weights of different classes of edges discussed above.
				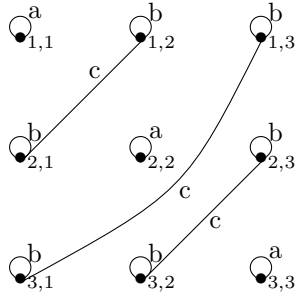
\begin{figure}
					\begin{center}
						\begin{tikzpicture}[scale = .8]
						\node at (0,4) {$\bullet_{1,1}$};
						\node at (2,4) {$\bullet_{1,2}$};
						\node at (4,4) {$\bullet_{1,3}$};
						\node at (0,2) {$\bullet_{2,1}$};
						\node at (2,2) {$\bullet_{2,2}$};
						\node at (4,2) {$\bullet_{2,3}$};
						\node at (0,0) {$\bullet_{3,1}$};
						\node at (2,0) {$\bullet_{3,2}$};
						\node at (4,0) {$\bullet_{3,3}$};
						
						\draw {[rounded corners] (-.24,4) -- (.01,4.25) -- (-.24, 4.5) -- (-.49, 4.25) -- (-.24,4)};
						\node at (0,4.5) {a};
						
						\draw {[rounded corners] (1.76,2) -- (2.01,2.25) -- (1.76, 2.5) -- (1.51, 2.25) -- (1.76,2)};
						\node at (2.01,2.5) {a};
						
						\draw {[rounded corners] (3.76,0) -- (4.01,0.25) -- (3.76, 0.5) -- (3.51, 0.25) -- (3.76,0)};
						\node at (4,0.5) {a};
						
						\draw {[rounded corners] (-.24,2) -- (.01,2.25) -- (-.24, 2.5) -- (-.49, 2.25) -- (-.24,2)};
						\node at (0,2.5) {b};
						
						\draw {[rounded corners] (-.24,0) -- (.01,0.25) -- (-.24, 0.5) -- (-.49, 0.25) -- (-.24,0)};
						\node at (0,0.5) {b};
						
						\draw {[rounded corners] (1.76,4) -- (2.01,4.25) -- (1.76, 4.5) -- (1.51, 4.25) -- (1.76,4)};
						\node at (2.01,4.5) {b};
						
						\draw {[rounded corners] (1.76,0) -- (2.01,0.25) -- (1.76, 0.5) -- (1.51, 0.25) -- (1.76,0)};
						\node at (2.01,0.5) {b};
						
						\draw {[rounded corners] (3.76,4) -- (4.01,4.25) -- (3.76, 4.5) -- (3.51, 4.25) -- (3.76,4)};
						\node at (4,4.5) {b};
						
						\draw {[rounded corners] (3.76,2) -- (4.01,2.25) -- (3.76, 2.5) -- (3.51, 2.25) -- (3.76,2)};
						\node at (4,2.5) {b};
						
						\draw (-.24,2) -- (1.76,4);
						\node at (1, 3.5) {c};
						
						\draw (1.76,0) -- (3.76,2);
						\node at (3, 1) {c};
						
						\draw {[rounded corners] (-.24,0) .. controls(2.5,1.5) .. (3.76,4)};
						\node at (2.5, 1.5) {c};
						
					\end{tikzpicture}
				\caption{Graph for $\rho_{x,3}$. Here $a = 2 + 2x$, $b = 3 - x$, and $c = 3x - 1$.}
				\label{werner_graph_1}
			\end{center}
		\end{figure} 
		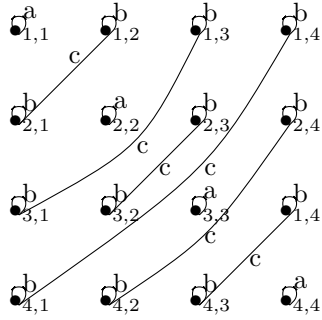
\begin{figure}
			\begin{center} 
				\begin{tikzpicture}[scale = .6]
					\node at (0,6) {$\bullet_{1,1}$};
					\node at (2,6) {$\bullet_{1,2}$};
					\node at (4,6) {$\bullet_{1,3}$};
					\node at (6,6) {$\bullet_{1,4}$};
					\node at (0,4) {$\bullet_{2,1}$};
					\node at (2,4) {$\bullet_{2,2}$};
					\node at (4,4) {$\bullet_{2,3}$};
					\node at (6,4) {$\bullet_{2,4}$};
					\node at (0,2) {$\bullet_{3,1}$};
					\node at (2,2) {$\bullet_{3,2}$};
					\node at (4,2) {$\bullet_{3,3}$};
					\node at (6,2) {$\bullet_{1,4}$};
					\node at (0,0) {$\bullet_{4,1}$};
					\node at (2,0) {$\bullet_{4,2}$};
					\node at (4,0) {$\bullet_{4,3}$};
					\node at (6,0) {$\bullet_{4,4}$};
					
					\draw {[rounded corners] (-.24,6) -- (.01,6.25) -- (-.24, 6.5) -- (-.49, 6.25) -- (-.24,6)};
					\node at (0,6.5) {a};
					
					\draw {[rounded corners] (1.76,4) -- (2.01,4.25) -- (1.76, 4.5) -- (1.51, 4.25) -- (1.76,4)};
					\node at (2.01,4.5) {a};
					
					\draw {[rounded corners] (3.76,2) -- (4.01,2.25) -- (3.76, 2.5) -- (3.51, 2.25) -- (3.76,2)};
					\node at (4,2.5) {a};
					
					\draw {[rounded corners] (5.76,0) -- (6.01,0.25) -- (5.76, 0.5) -- (5.51, 0.25) -- (5.76,0)};
					\node at (6,0.5) {a};
					
					\draw {[rounded corners] (-.24,4) -- (.01,4.25) -- (-.24, 4.5) -- (-.49, 4.25) -- (-.24,4)};
					\node at (0,4.5) {b};
					
					\draw {[rounded corners] (1.76,6) -- (2.01,6.25) -- (1.76, 6.5) -- (1.51, 6.25) -- (1.76,6)};
					\node at (2.01,6.5) {b};
					
					\draw {[rounded corners] (1.76,2) -- (2.01,2.25) -- (1.76, 2.5) -- (1.51, 2.25) -- (1.76,2)};
					\node at (2.01,2.5) {b};
					
					\draw {[rounded corners] (3.76,0) -- (4.01,0.25) -- (3.76, 0.5) -- (3.51, 0.25) -- (3.76,0)};
					\node at (4,0.5) {b};
					
					\draw {[rounded corners] (-.24,2) -- (.01,2.25) -- (-.24, 2.5) -- (-.49, 2.25) -- (-.24,2)};
					\node at (0,2.5) {b};
					
					\draw {[rounded corners] (-.24,0) -- (.01,0.25) -- (-.24, 0.5) -- (-.49, 0.25) -- (-.24,0)};
					\node at (0,0.5) {b};
					
					\draw {[rounded corners] (1.76,0) -- (2.01,0.25) -- (1.76, 0.5) -- (1.51, 0.25) -- (1.76,0)};
					\node at (2.01,0.5) {b};
					
					\draw {[rounded corners] (3.76,4) -- (4.01,4.25) -- (3.76, 4.5) -- (3.51, 4.25) -- (3.76,4)};
					\node at (4,4.5) {b};
					
					\draw {[rounded corners] (5.76,2) -- (6.01,2.25) -- (5.76, 2.5) -- (5.51, 2.25) -- (5.76,2)};
					\node at (6,2.5) {b};
					
					\draw {[rounded corners] (5.76,4) -- (6.01,4.25) -- (5.76, 4.5) -- (5.51, 4.25) -- (5.76,4)};
					\node at (6,4.5) {b};
					
					\draw {[rounded corners] (5.76,6) -- (6.01,6.25) -- (5.76, 6.5) -- (5.51, 6.25) -- (5.76,6)};
					\node at (6,6.5) {b};
					
					\draw {[rounded corners] (3.76,6) -- (4.01,6.25) -- (3.76, 6.5) -- (3.51, 6.25) -- (3.76,6)};
					\node at (4,6.5) {b};
					
					\draw (-.24,4) -- (1.76,6);
					\node at (1, 5.5) {c};
					
					\draw (1.76,2) -- (3.76,4);
					\node at (3, 3) {c};
					
					\draw {[rounded corners] (-.24,2) .. controls(2.5,3.5) .. (3.76,6)};
					\node at (2.5, 3.5) {c};
					
					\draw {[rounded corners] (1.76,0) .. controls(4,1.5) .. (5.76,4)};
					\node at (4, 1.5) {c};
					
					\draw (3.76,0) -- (5.76,2);
					\node at (5, 1) {c};
					
					\draw {[rounded corners] (-.24,0) .. controls(4,3) .. (5.76,6)};
					\node at (4, 3) {c};
				\end{tikzpicture}
				\caption{Graph for $\rho_{x,4}$. Here, $a = 3x + 3$, $b = 4 - x$, and $c = 4x - 1$.}
				\label{werner_graph_2}
			\end{center}
		\end{figure}
		\end{example}
		
		\begin{theorem}\label{werner_state_discord_statement}
			Graph Laplacian Werner states have non-zero quantum discord except for certain discrete values of $x$.
		\end{theorem}
		
		\begin{proof} 
			Note that, for all $x$ there is only an edge $(v_{\mu,i}, v_{i, \mu})$ in the subgraph $\langle C_\mu, C_i \rangle$ where $\mu \neq i$. After the lemma \ref{normality} we have shown such type of graphs cannot fulfill normality condition.
			
			As an example, consider the subdigraph $\langle C_1, C_2 \rangle$ of $\rho_{x,3}$ depicted in figure \ref{werner_proof_graph}. 
			\begin{figure}
				\begin{center}
					\begin{tikzpicture}
					\node at (0,1) {$\bullet_{1,1}$};
					\node at (2,1) {$\bullet_{1,2}$};
					\node at (4,1) {$\bullet_{1,3}$};
					\node at (0,0) {$\bullet_{2,1}$};
					\node at (2,0) {$\bullet_{2,2}$};
					\node at (4,0) {$\bullet_{2,3}$};
					\draw (1.8,1.05) -- (-.2,.05);
					\node at (1.8, .5) {$(3x - 1)$};
					\end{tikzpicture}
				\end{center}
				\caption{Subgraph $\langle C_1, C_2 \rangle$ of the graph $\rho_{x,3}$ drawn in figure \ref{werner_graph_1}}.
				\label{werner_proof_graph} 
			\end{figure}
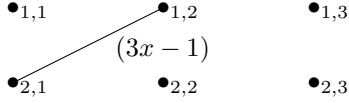
			There is only one edge $(v_{1, 2}, v_{2,1})$ with weight $(3x - 1)$ between two clusters $C_1$ and $C_2$. The edge weight is non-zero when $x \neq \frac{1}{3}$. Note that, $w(v_{12}, v_{21}) w(v_{21}, v_{12}) = (3x - 1)^2$ but $w(v_{22}, v_{12})w(v_{12}, v_{12}) = 0$ as $v_{22}$ is an isolated vertex. In this case, the graph $\langle C_1, C_2 \rangle$ fulfills the normality condition if and only if $x = \frac{1}{3}$.
			
			Thus the normality condition of the theorem \ref{math_thm} is violated except for some parameter values. Hence, graph Laplacian Werner states have non-zero quantum discord except for some specific values of $x$.
		\end{proof}

		\subsection{Isotropic state}
			
			An isotropic state $\rho_{d,x}$ acting on $\mathcal{H}^{(d)} \otimes \mathcal{H}^{(d)}$ is defined by,
			\begin{equation}\label{isotropic_equation}
			\rho_{d,x} = \frac{d^2}{d^2 - 1}\left[\frac{(1 - F)}{d^2} I + \left(F - \frac{1}{d^2} \right)P \right],
			\end{equation}
			where $F \in [0,1]$ is the fidelity of the quantum state and $P = \ket{\psi}\bra{\psi}$ where $\ket{\psi} = \frac{1}{\sqrt{d}}\sum_{i}\ket{i_a}\ket{i_b}$, the maximally entangled state in dimension $d$. Discord of isotropic state is studied in \cite{luo2010geometric, guo2016non}. Considering diagonal and off-diagonal terms we may conclude that an isotropic quantum state is a graph Laplacian state provided 
			\begin{equation}
			(d - 1)\left| F - \frac{1}{d^2}\right| \le \frac{d^2 - 1}{d^2}F.
			\end{equation}
			Putting $d = 2, 3, 4$ in the above equation we get, $\frac{1}{7} \le F \le 1$, $\frac{1}{13} \le F \le \frac{1}{5}$, $\frac{1}{11} \le F \le \frac{1}{21}$, respectively.
			
			As $\rho_{d,x}$ acts on $\mathcal{H}^{(d)} \otimes \mathcal{H}^{(d)}$, we represent the vertex set into $d$ clusters $C_\mu: \mu = 1, 2, \dots d$ with $C_\mu = \{v_{\mu 1}, v_{\mu 2}, \dots v_{\mu n}\}$. We observe that a graph representing an isotropic state has the following properties.
			\begin{enumerate}
				\item
				The diagonal vertices $v_{1,1}, v_{22}, \dots v_{dd}$ form a complete graph which consists of all non-loop edges of the graph. Weight of these edges are $F - \frac{1}{d^2}$.
				\item
				The loop weight of the non-diagonal vertices is $\frac{1 - F}{d^2}$.
				\item
				The loop weight of the diagonal vertices are given by $\frac{d^2 - 1}{d^2}F$.
			\end{enumerate}
			\begin{example}
				Graph representations of the isotropic state $\rho_{d,x}$ for $d = 2, 3, 4$ are depicted in the figure \ref{isotropic_graphs_2}, \ref{isotropic_graphs_3}, and \ref{isotropic_graphs_4}. In the picture, all the edges and loops are weighted as described above.
				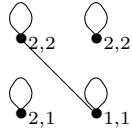
\begin{figure}
					\begin{center} 
						\begin{tikzpicture}
						\node at (0,0) {$\bullet_{2,1}$};
						\node at (0,1) {$\bullet_{2,2}$};
						\node at (1,0) {$\bullet_{1,1}$};
						\node at (1,1) {$\bullet_{2,2}$};
						\draw (-.25,1.1) -- (.75,.1);
						\draw {[rounded corners] (-.2, .1) -- (-.4, .3) -- (-.2, .6) -- (0, .3) -- (-.2, .1)};
						\draw {[rounded corners] (-.2, 1.1) -- (-.4, 1.3) -- (-.2, 1.6) -- (0, 1.3) -- (-.2, 1.1)};
						\draw {[rounded corners] (.8, .1) -- (.6, .3) -- (.8, .6) -- (1, .3) -- (.8, .1)};
						\draw {[rounded corners] (.8, 1.1) -- (.6, 1.3) -- (.8, 1.6) -- (1, 1.3) -- (.8, 1.1)};
						\end{tikzpicture}
						\caption{Isotropic state $\rho_{2,x}$.}
						\label{isotropic_graphs_2}
					\end{center}
				\end{figure}
			 
				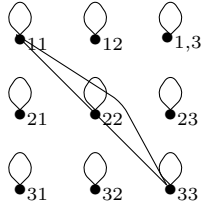
\begin{figure}
					\begin{center} 
						\begin{tikzpicture}
							\node at (0,0) {$\bullet_{31}$};
							\node at (1,0) {$\bullet_{32}$};
							\node at (2,0) {$\bullet_{33}$};
							\node at (0,1) {$\bullet_{21}$};
							\node at (1,1) {$\bullet_{22}$};
							\node at (2,1) {$\bullet_{23}$};
							\node at (0,2) {$\bullet_{11}$};
							\node at (1,2) {$\bullet_{12}$};
							\node at (2,2) {$\bullet_{1,3}$};
							\draw (-.2,2.1) -- (.8,1.1) -- (1.8, .1);
							\draw {[rounded corners] (-.2,2.1) -- (1.2, 1.2) -- (1.8, .1)};
							\draw {[rounded corners] (-.15, .1) -- (-.35, .3) -- (-.15, .6) -- (.05, .3) -- (-.15, .1)};
							\draw {[rounded corners] (.85, .1) -- (.65, .3) -- (.85, .6) -- (1.05, .3) -- (.85, .1)};
							\draw {[rounded corners] (1.85, .1) -- (1.65, .3) -- (1.85, .6) -- (2.05, .3) -- (1.85, .1)};
							\draw {[rounded corners] (-.15, 1.1) -- (-.35, 1.3) -- (-.15, 1.6) -- (.05, 1.3) -- (-.15, 1.1)};
							\draw {[rounded corners] (.85, 1.1) -- (.65, 1.3) -- (.85, 1.6) -- (1.05, 1.3) -- (.85, 1.1)};
							\draw {[rounded corners] (1.85, 1.1) -- (1.65, 1.3) -- (1.85, 1.6) -- (2.05, 1.3) -- (1.85, 1.1)};
							\draw {[rounded corners] (-.15, 2.1) -- (-.35, 2.3) -- (-.15, 2.6) -- (.05, 2.3) -- (-.15, 2.1)};
							\draw {[rounded corners] (.85, 2.1) -- (.65, 2.3) -- (.85, 2.6) -- (1.05, 2.3) -- (.85, 2.1)};
							\draw {[rounded corners] (1.85, 2.1) -- (1.65, 2.3) -- (1.85, 2.6) -- (2.05, 2.3) -- (1.85, 2.1)};
						\end{tikzpicture}
						\caption{Isotropic state $\rho_{3,x}$} 
						\label{isotropic_graphs_3}
					\end{center} 
				\end{figure}
			
				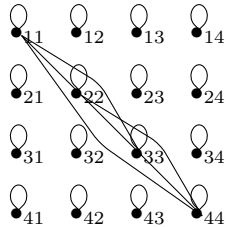
\begin{figure}
					\begin{center} 
						\begin{tikzpicture}[scale = .8]
							\node at (0,0) {$\bullet_{41}$};
							\node at (1,0) {$\bullet_{42}$};
							\node at (2,0) {$\bullet_{43}$};
							\node at (3,0) {$\bullet_{44}$};
							\node at (0,1) {$\bullet_{31}$};
							\node at (1,1) {$\bullet_{32}$};
							\node at (2,1) {$\bullet_{33}$};
							\node at (3,1) {$\bullet_{34}$};
							\node at (0,2) {$\bullet_{21}$};
							\node at (1,2) {$\bullet_{22}$};
							\node at (2,2) {$\bullet_{23}$};
							\node at (3,2) {$\bullet_{24}$};
							\node at (0,3) {$\bullet_{11}$};
							\node at (1,3) {$\bullet_{12}$};
							\node at (2,3) {$\bullet_{13}$};
							\node at (3,3) {$\bullet_{14}$};
							\draw (-.1, 3) -- (.9,2) -- (1.9,1) -- (2.9,0);
							\draw {[rounded corners] (-.1, 3) -- (1.2, 2.2) -- (1.9,1)};
							\draw {[rounded corners] (.9,2) -- (2.2, 1.2) -- (2.9,0)};
							\draw {[rounded corners] (-.1, 3) -- (1.2, 1.2) -- (2.9,0)};
							\draw {[rounded corners] (-.15, .1) -- (-.35, .3) -- (-.15, .6) -- (.05, .3) -- (-.15, .1)};
							\draw {[rounded corners] (.85, .1) -- (.65, .3) -- (.85, .6) -- (1.05, .3) -- (.85, .1)};
							\draw {[rounded corners] (1.85, .1) -- (1.65, .3) -- (1.85, .6) -- (2.05, .3) -- (1.85, .1)};
							\draw {[rounded corners] (2.85, .1) -- (2.65, .3) -- (2.85, .6) -- (3.05, .3) -- (2.85, .1)};
							\draw {[rounded corners] (-.15, 1.1) -- (-.35, 1.3) -- (-.15, 1.6) -- (.05, 1.3) -- (-.15, 1.1)};
							\draw {[rounded corners] (.85, 1.1) -- (.65, 1.3) -- (.85, 1.6) -- (1.05, 1.3) -- (.85, 1.1)};
							\draw {[rounded corners] (1.85, 1.1) -- (1.65, 1.3) -- (1.85, 1.6) -- (2.05, 1.3) -- (1.85, 1.1)};
							\draw {[rounded corners] (2.85, 1.1) -- (2.65, 1.3) -- (2.85, 1.6) -- (3.05, 1.3) -- (2.85, 1.1)};
							\draw {[rounded corners] (-.15, 2.1) -- (-.35, 2.3) -- (-.15, 2.6) -- (.05, 2.3) -- (-.15, 2.1)};
							\draw {[rounded corners] (.85, 2.1) -- (.65, 2.3) -- (.85, 2.6) -- (1.05, 2.3) -- (.85, 2.1)};
							\draw {[rounded corners] (1.85, 2.1) -- (1.65, 2.3) -- (1.85, 2.6) -- (2.05, 2.3) -- (1.85, 2.1)};
							\draw {[rounded corners] (2.85, 2.1) -- (2.65, 2.3) -- (2.85, 2.6) -- (3.05, 2.3) -- (2.85, 2.1)};
							\draw {[rounded corners] (-.15, 3.1) -- (-.35, 3.3) -- (-.15, 3.6) -- (.05, 3.3) -- (-.15, 3.1)};
							\draw {[rounded corners] (.85, 3.1) -- (.65, 3.3) -- (.85, 3.6) -- (1.05, 3.3) -- (.85, 3.1)};
							\draw {[rounded corners] (1.85, 3.1) -- (1.65, 3.3) -- (1.85, 3.6) -- (2.05, 3.3) -- (1.85, 3.1)};
							\draw {[rounded corners] (2.85, 3.1) -- (2.65, 3.3) -- (2.85, 3.6) -- (3.05, 3.3) -- (2.85, 3.1)};
						\end{tikzpicture}
					\caption{Isotropic state $\rho_{4,x}$}
					\label{isotropic_graphs_4}   
					\end{center} 
				\end{figure}
			\end{example}
			\begin{theorem}
				Graph Laplacian isotropic states have non-zero quantum discord except for certain discrete values of $F$.
			\end{theorem}
			\begin{proof}
				From the graph structure of the state $\rho$, Eq. (\ref{isotropic_equation}), we see that the family of subgraphs $\{\langle C_\mu, C_\nu\rangle\}$ do not satisfy the commutativity and normality criterion, except some specific edge weights. 
				
				As an example consider the subgraph $\langle C_1, C_2 \rangle$ of the graph $\rho_{3,x}$ depicted in the figure \ref{proof_isotropic}. The subgraph $\langle C_1, C_2 \rangle$ also breaks the normality condition for all non-zero edge weights due to reasons similar to those stated in theorem \ref{werner_state_discord_statement}.
				\begin{figure}
					\begin{center}
						\begin{tikzpicture}
						\node at (0,1) {$\bullet_{1,1}$};
						\node at (2,1) {$\bullet_{1,2}$};
						\node at (4,1) {$\bullet_{1,3}$};
						\node at (0,0) {$\bullet_{2,1}$};
						\node at (2,0) {$\bullet_{2,2}$};
						\node at (4,0) {$\bullet_{2,3}$};
						\draw (0,.9) -- (1.8, .1);
						\node at (1.8, .5) {$(F - \frac{1}{d^2})$};
						\end{tikzpicture}
						\caption{Subgraph $\langle C_1, C_2 \rangle$ of $\rho_{3,x}$ depicted in the figure \ref{isotropic_graphs_3}}
						\label{proof_isotropic}
					\end{center}
				\end{figure}
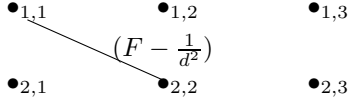
				
				Thus, we may conclude that graph Laplacian isotropic states have non-zero quantum discord except for some specific values of $F$.
			\end{proof}
			
		\subsection{X state}
			
			The $X$-state is well known in quantum information theory due to the specific structure of its density matrix. Discord of some classes of $2$-qubit $X$-states have been studied in the literature \cite{ali2010quantum, sabapathy2013quantum}. Here, we consider graph Laplacian (definition \ref{graphical_state_definition}) $X$-states acting on $\mathcal{H}^{(m)} \otimes \mathcal{H}^{(n)}$. Hence, as before the vertex set of the corresponding digraph has $m$ clusters $C_\mu, \mu = 1, 2, \dots m$, each containing $n$ vertices. The distribution of the non-zero elements in the density matrices suggest that the edge set has the following combinatorial characteristics:
			\begin{enumerate}
				\item 
				If the bipartite subgraph $\langle C_{\mu}, C_{\nu} \rangle$ is non-empty then all the edges are of the form $(v_{\mu k}, v_{\nu (n + 1 - k)})$ for $k = 1, 2, \dots n$.
				\item
				There is only one non-empty subgraph $\langle C_{\alpha} \rangle$ with edges of the form $(v_{\alpha k}, v_{\alpha (n + 1 - k)})$ for $k = 1, 2, \dots n$.
			\end{enumerate}
			Conversely if the edge set of any graph follows the above two properties the corresponding quantum states will be classified as an $X$ state.
			
			\begin{example}
				Some of the graphs of $X$ states without edge weights and directions are depicted in the figures \ref{x_state_1} and \ref{x_state_2}.
				\begin{figure}
					\begin{center}
						\begin{tikzpicture}
							\node at (0,0) {$\bullet_{1,1}$};
							\node at (0,1) {$\bullet_{2,1}$};
							\node at (0,2) {$\bullet_{3,1}$};
							\node at (1,0) {$\bullet_{1,2}$};
							\node at (1,1) {$\bullet_{2,2}$};
							\node at (1,2) {$\bullet_{3,2}$};
							\node at (2,0) {$\bullet_{1,3}$};
							\node at (2,1) {$\bullet_{2,3}$};
							\node at (2,2) {$\bullet_{3,3}$};
							\draw (-.15,.1) .. controls(1.2,.8).. (1.85,2.1);
							\draw (-.15,2.1) .. controls(.7, .8).. (1.85,.1);
							\draw (-.15,1.1) .. controls(1, 1.5).. (1.85,1.1);
							\draw (.85, .1) .. controls(1.5, 1).. (.85,2); 
						\end{tikzpicture}
					\end{center} 
					\caption{$X$ state acting on $\mathcal{H}^3 \otimes \mathcal{H}^3$.}
					\label{x_state_1}
				\end{figure}
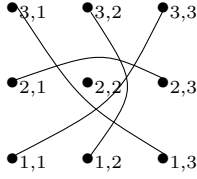

				\begin{figure}
					\begin{center} 
						\begin{tikzpicture}[scale = .6]
							\node at (0,0) {$\bullet_{1,1}$};
							\node at (0,2) {$\bullet_{2,1}$};
							\node at (0,4) {$\bullet_{3,1}$};
							\node at (2,0) {$\bullet_{1,2}$};
							\node at (2,2) {$\bullet_{2,2}$};
							\node at (2,4) {$\bullet_{3,2}$};
							\node at (4,0) {$\bullet_{1,3}$};
							\node at (4,2) {$\bullet_{2,3}$};
							\node at (4,4) {$\bullet_{3,3}$};
							\node at (6,0) {$\bullet_{1,4}$};
							\node at (6,2) {$\bullet_{2,4}$};
							\node at (6,4) {$\bullet_{3,4}$};
							\draw (-.2, .1) -- (5.8, 4.1);
							\draw (1.8, .1) -- (3.8, 4.1);
							\draw (3.8, .1) -- (1.8, 4.1);
							\draw (5.8, .1) -- (-.2, 4.1);
							\draw (1.8, 2.1) -- (3.8, 2.1);
							\draw (-.2, 2.1) .. controls(3, 2.5).. (5.8,  2.1);
						\end{tikzpicture}
					\end{center} 
					\caption{$X$ state acting on $\mathcal{H}^3 \otimes \mathcal{H}^4$.}
					\label{x_state_2}
				\end{figure}
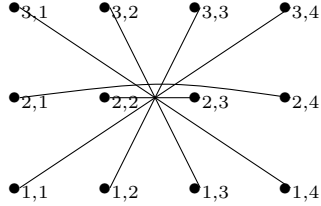
			\end{example}
			
			\begin{theorem}\label{x_state_condition}
				A graph Laplacian $X$ state acting on $\mathcal{H}^{(m)} \otimes \mathcal{H}^{(n)}$ has zero quantum discord if and only if the following conditions are satisfied:
				\begin{enumerate}
					\item
					Any two non-empty subdigraphs of the form $\langle C_{\mu}, C_{\nu} \rangle$ are equal.
					\item 
					Degree of the vertices of $C_\mu$ will fulfil $d_{\mu i} = d_{\mu (n + 1 - i)}$ for $i = 1, 2, \dots n$.
				\end{enumerate}
			\end{theorem}
			
			\begin{proof}
				Recall that if two subdigraphs $\langle C_{\mu}, C_{\nu} \rangle$ and $\langle C_{\alpha}, C_{\beta} \rangle$ are equal, then the commutativity condition is satisfied. Also, if any one of them is empty, the commutativity condition is again satisfied. Now we consider the subgraphs $\langle C_{\mu}, C_{\nu} \rangle$ and $\langle C_{\alpha} \rangle$. When any one of them is an empty graph the commutativity condition is satisfied trivially. There is only one non-empty subgraph of the form $\langle C_{\alpha} \rangle$. Using corollary \ref{commutativity_1} we may verify that the non-empty graphs $\langle C_{\mu}, C_{\nu} \rangle$ and $\langle C_{\alpha} \rangle$ are commutative. Also using lemma \ref{normality} we can show that subgraphs $\langle C_{\mu}, C_{\nu} \rangle$ and $\langle C_{\alpha} \rangle$ satisfy the conditions for being normal. Last, we shall check the degree condition, 
				\begin{equation}
				w(v_{\mu i}, v_{\nu j}) (d_{\alpha i} - d_{\alpha j}) = 0.
				\end{equation}
				As $\langle C_{\mu}, C_{\nu} \rangle$ is non-empty, we have $w(v_{\mu i}, v_{\mu (n + 1 - i)}) \neq 0$ for some $i = 1, 2, \dots n$. For those specific values of $i$ we have,
				\begin{equation}
				w(v_{\mu i}, v_{\nu (n + 1 - i)}) (d_{\alpha i} - d_{\alpha (n + 1 - i)}) = 0.
				\end{equation}
				As $w(v_{\mu i}, v_{\nu (n + 1 - i)}) \neq 0$, we have $d_{\alpha i} = d_{\alpha (n + 1 - i)}$ for $i = 1, 2, \dots n$.
			\end{proof}
			
			\begin{example}
				Consider the following graph:
				$$
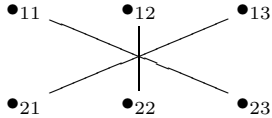
\xymatrix{\bullet_{11} \ar@{-}[drr] & \bullet_{12} \ar@{-}[d] & \bullet_{13} \ar@{-}[dll] \\ \bullet_{21} & \bullet_{22} & \bullet_{23}}$$
				Edge weight $w(v_{11}, v_{13}) = 2$ and for the other two edges, weight is $1$. Here number of clusters $m = 2$ and number of vertices in each cluster $n = 3$. The corresponding quantum state is given by the density matrix,
				\begin{equation}
				\rho(G) = \frac{1}{8}\begin{bmatrix} 2 & 0 & 0 & 0 & 0 & 2 \\ 0 & 1 & 0 & 0 & 1 & 0 \\ 0 & 0 & 1 & 1 & 0 & 0 \\ 0 & 0 & 1 & 1 & 0 & 0\\ 0 & 1 & 0 & 0 & 1 & 0 \\ 2 & 0 & 0 & 0 & 0 & 2 \end{bmatrix}, 
				\end{equation}
				which lies in $\mathcal{H}^{(2)} \otimes \mathcal{H}^{(3)}$. Degree of the vertices are: $d(v_{11}) = 2, d(v_{12}) = 1, d(v_{13}) = 1, d(v_{11}) = 1, d(v_{22}) = 1, d(v_{23}) = 2$. According to the second condition of the above theorem, for zero discord $d(v_{11}) = d(v_{13})$ which is not fulfilled in this case. Hence, the corresponding quantum state has non-zero discord.
			\end{example}
			
			In general the density matrix of a two-qubit $X$ state is given by,
			\begin{equation}\label{two_qubbit_x_state}
			\rho = \begin{bmatrix} \rho_{11} & 0 & 0 & \rho_{14} \\ 0 & \rho_{22} & \rho_{23} & 0 \\ 0 & \rho_{32} & \rho_{33} & 0 \\ \rho_{41} & 0 & 0 & \rho_{44} \end{bmatrix}.
			\end{equation}
			It must be a Hermitian, positive semidefinite, trace one matrix. To satisfy Hermiticity, $\rho_{41} = \overline{\rho_{14}}, \rho_{32} = \overline{\rho_{23}}$ and $\rho_{ii}$ are real for all $i$. The positivity condition requires that $\rho_{22}\rho_{33} \geq |\rho_{23}|^2$ and $\rho_{11}\rho_{44} \geq |\rho_{14}|^2$. Also, for unit trace $\sum_{i = 1}^4 \rho_{ii} = 1$. Lemma \ref{graphical_states} implies that $\rho$ represents a graph Laplacian state if and only if 
			\begin{equation}\label{x_state_graphical_conditions}
			\rho_{11} \geq |\rho_{14}|, \rho_{22} \geq |\rho_{23}|, \rho_{33} \geq |\rho_{32}| ~\text{and}~ \rho_{44} \geq |\rho_{41}|.
			\end{equation} 
			A graph with four vertices distributed into two clusters, each containing two vertices represent $\rho$ as a graph Laplacian state. Discord of the state depends on the edge distribution in the graph. For simplicity, let the graph have no loops. Then the equation (\ref{x_state_graphical_conditions}) simplifies to
			\begin{equation}
			\rho_{11} = |\rho_{14}|, \rho_{22} = |\rho_{23}|, \rho_{33} = |\rho_{32}| ~\text{and}~ \rho_{44} = |\rho_{41}|.
			\end{equation}
			Combining this with the positivity conditions we get,
			\begin{equation}
			\begin{split}
			& \rho_{11} = |\rho_{14}| = |\rho_{41}| = \rho_{44} = a \\
			& \rho_{22} = |\rho_{23}| = |\rho_{32}| = \rho_{33} = b
			\end{split}
			\end{equation}
			for some real numbers $a$ and $b$. A graph satisfying the above condition is
			$$
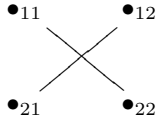
\xymatrix{\bullet_{11} \ar@{-}[rd] & \bullet_{12} \ar@{-}[ld] \\ \bullet_{21} & \bullet_{22}}$$
			Here, weights of $(v_{11}, v_{22})$ and $(v_{12}, v_{21})$ are $a$ and $b$, respectively. Degree of the vertices are given by $d(v_{11}) = a, d(v_{12}) = b, d(v_{21}) = b$ and $d(v_{22}) = a$. Now by theorem \ref{x_state_condition}, the corresponding quantum state has zero discord if and only if $a = b$. In all other cases $\rho$ has non-zero discord. Further, we know that a two qubit $X$-state is entangled if and only if either $\rho_{22} \rho_{33} < |\rho_{14}|^2$ or $\rho_{11} \rho_{44} < |\rho_{23}|^2$. From this we can conclude that if $a = b$ then entanglement is also zero. This coincides with the results in \cite{ali2010quantum}.
			
			As an important example of the above considered general two qubit $X$ state, we take up the two qubit Werner state, given by
			\begin{equation}
			\rho = a\ket{\psi^-}\bra{\psi^-} + \frac{1 - a}{4}I,
			\end{equation}
			where $\ket{\psi^-} = \frac{1}{\sqrt{2}}(\ket{01} - \ket{10})$ and $0 \leq a \leq 1$. The density matrix in expanded from is,
			\begin{equation}
			\rho = \begin{bmatrix} \frac{1 - a}{4} & 0 & 0 & 0 \\ 0 & \frac{1 + a}{4} & \frac{-a}{2} & 0 \\ 0 & \frac{-a}{2} & \frac{1 + a}{4} & 0 \\ 0 & 0 & 0 & \frac{1-a}{4} \end{bmatrix}.
			\end{equation}
			As $a \leq 1$, clearly $\frac{1 - a}{4} \geq 0$ and $\frac{1 + a}{4} \geq \frac{a}{2}$. Therefore, $\rho$ represents a graph Laplacian quantum state for all values of $a$. The graph representing a two qubit Werner state is
			$$
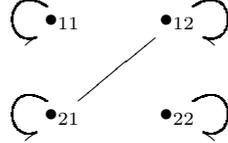
\xymatrix{\bullet_{11} \ar@(ul, dl) & \bullet_{12} \ar@(ur, dr) \ar@{-}[dl] \\ \bullet_{21} \ar@(ul, dl) & \bullet_{22} \ar@(ur, dr) }$$
			including loop weight $\frac{1 - a}{8}$, and edge weight $\frac{a}{2}$. Therefore, degree of the vertices are $d(v_{11}) = \frac{1 - a}{8}, d(v_{12}) = \frac{1 - 3a}{8}, d(v_{21}) = \frac{1 - 3a}{8}$ and $d(v_{22}) = \frac{1 - a}{8}$. For zero discord, we need $\frac{1 - a}{8} = \frac{1 - 3a}{8}$, which implies that $a = 0$. This is consistent with the results in \cite{ali2010quantum}.

			\subsection{Graph Laplacian quantum states corresponding to simple graphs}
			
			A simple graph $G$ satisfies the basis assumptions stated in section 2. Given any edge $(i,j)$ of a simple graph, the edge weight $w(i,j) = w(j,i) = 1$. Also, a simple graph has no loop, that is, $(i,i) \notin E(G)$ for all vertices $i$. These assumptions simplify the conditions of theorem 1. A detailed description on quantum discord of graph Laplacian quantum states arising from simple graphs is available in \cite{dutta2017quantum}. Here we present a specific example. Consider the following graph:
			$$
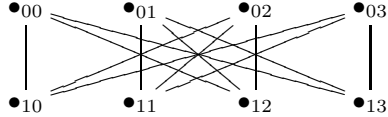
\xymatrix{\bullet_{00} \ar@{-}[d] \ar@{-}[drr] \ar@{-}[drrr] & \bullet_{01} \ar@{-}[d] \ar@{-}[dr] \ar@{-}[drr] & \bullet_{02} \ar@{-}[d] \ar@{-}[dl] \ar@{-}[dll] & \bullet_{03} \ar@{-}[d] \ar@{-}[dll] \ar@{-}[dlll] \\ \bullet_{10} & \bullet_{11} & \bullet_{12} & \bullet_{13}}$$
			
			It has two clusters $C_0$ and $C_1$, each containing $4$ vertices. Note that, there is only one bipartite subgraph $\langle C_0, C_1 \rangle$ in the above graph. Also, degree of every vertex is three. The density matrices corresponding to this graph are:
			\begin{equation}
			\begin{split} 
			& \rho_l(G) = \frac{1}{24}\begin{bmatrix} 3 & 0 & 0 & 0 & -1 & 0 & -1 & -1 \\ 0 & 3 & 0 & 0 & 0 & -1 & -1 & -1 \\ 0 & 0 & 3 & 0 & -1 & -1 & -1 & 0 \\ 0 & 0 & 0 & 3 & -1 & -1 & 0 & -1 \\ -1 & 0 & -1 & -1 & 3 & 0 & 0 & 0 \\ 0 & -1 & -1 & -1 & 0 & 3 & 0 & 0 \\ -1 & -1 & -1 & 0 & 0 & 0 & 3 & 0 \\ -1 & -1 & 0 & -1 & 0 & 0 & 0 & 3 \end{bmatrix},\\
			\text{and}~ & \rho_q(G) = \frac{1}{24}\begin{bmatrix} 3 & 0 & 0 & 0 & 1 & 0 & 1 & 1 \\ 0 & 3 & 0 & 0 & 0 & 1 & 1 & 1 \\ 0 & 0 & 3 & 0 & 1 & 1 & 1 & 0 \\ 0 & 0 & 0 & 3 & 1 & 1 & 0 & 1 \\ 1 & 0 & 1 & 1 & 3 & 0 & 0 & 0 \\ 0 & 1 & 1 & 1 & 0 & 3 & 0 & 0 \\ 1 & 1 & 1 & 0 & 0 & 0 & 3 & 0 \\ 1 & 1 & 0 & 1 & 0 & 0 & 0 & 3 \end{bmatrix}.
			\end{split} 
			\end{equation}
			This graph satisfies all the conditions of theorem 1. Therefore, the mixed quantum states corresponding to this graph have zero discord.

	\section{Conclusions}
	
		This work extends the study of quantum discord of graph Laplacian states arising from simple graphs to that of graph Laplacian states arising from weighted digraphs. This covers a wider set of quantum states, including mixed states, represented by graphs. We establish that a quantum state is a graph Laplacian state if and only if its density matrix is diagonally dominant. We study the nature of discord in a number of well known quantum states, for example, two qubit graph Laplacian states including Bell states, Werner, Isotropic, and $X$ states. It emerges that the nature of quantum discord can be visualized graphically. All Werner and isotropic states are seen to have nonzero quantum discord, except for certain discrete values of their parameters. Also, an $X$ state has zero discord if and only if the underlined graph follows a particular degree sequence which has been used for analysing discord of two qubit $X$-states. The following problems may be attempted in future:
		
		\begin{enumerate}
			\item 
				We have shown that there are Isotropic states that are not graph Laplacian states and hence it would be worthwhile to develop graph theoretic representation of these states.
			\item
				Very recently a discord based quantum cryptography has been discussed \cite{pirandola2014quantum}. Also, there is a lot of interest in applying mixed quantum states to different tasks in quantum information \cite{rosmanis2006mixed}. These two possibilities open up the scope to design quantum cryptographic protocols based on graph Laplacian quantum states.
			\item
				In Nuclear magnetic resonance (NMR) based quantum computation use is made of pseudo-pure quantum states \cite{cory1997ensemble}. These are mixed states. Graph Laplacian states are also mixed in general. Constructing pseudo-pure states with Graph Laplacian states will be a worthwhile task in this direction.
			\item 
				The idea of quantum discord is further generalized when the von-Neumann entropy is replaced by Sharma-Mittal, R{\'e}nyi, and Tsallis entropy \cite{mazumdar2018sharma}. Constructing a combinatorial aspect of these new discords will also be a fascinating task.
		\end{enumerate}
		
		This work, which relies on the interface between graph theory and quantum mechanics would be useful for investigation of discord of a bigger class of quantum states with a corresponding pictorial description.

	\section*{Appendix}
		
		Proof of Lemma \ref{commutativity}
		\begin{proof} 
			Commutativity $AB = BA$ holds if and only if $(AB)_{ij} = (BA)_{ij}$ for all $i, j$ with $1 \le i,j \le n$. Note that, $a_{i k} = w(v_{\mu i}, v_{\nu k})$ and $b_{k j} = w(v_{\alpha k}, v_{\beta j})$. Now applying equation (\ref{inner_product}) we get,
			\begin{equation}\label{commutativity_main}
			\begin{split} 
			(AB)_{ij} & = \sum_{k = 1}^n a_{ik}b_{kj} = \braket{a_{i*}| b_{*j}} = \sum_k w(v_{\mu i}, v_{\nu k}) w(v_{\alpha k}, v_{\beta j}) : k \in \nbd(v_{\mu i}) \cap \nbd(v_{\beta j}),\\
			(BA)_{ij} & = \sum_{k = 1}^n b_{ik}a_{kj} = \braket{b_{i*}| a_{*j}} = \sum_k w(v_{\alpha i}, v_{\beta k}) w(v_{\mu k}, v_{\nu j}) : k \in \nbd(v_{\alpha i}) \cap \nbd(v_{\nu j}).
			\end{split}
			\end{equation}
		\end{proof}
	
		Proof of corollary \ref{commutativity_1}		
		\begin{proof}
			We have already justified that, $\supp(a_{i*}) = \nbd_{\text{\~{A}}}(v_{\mu i})$ and $\supp(a_{*i}) = \nbd_{\text{\~{A}}}(v_{\mu i})$, for all $i = 1, 2, \dots n$. The matrix $A$ commutes with $B$, if and only if the product $\langle a_{i*}, b_{*j} \rangle = \langle b_{i*}, a_{*j} \rangle$ for all $i$, and $j$. Applying the symmetry of $A$, we get, $\langle a_{i*}, b_{*j} \rangle = \langle a_{j*}, b_{i*} \rangle$. Using the graph theoretic convention, we get the desired result.
		\end{proof}
	
		Proof of corollary \ref{commutativity_2} 		
		\begin{proof}
			The proof follows from the above Corollary by choosing $\alpha = \beta = \nu$.
		\end{proof}
		
		Proof of corollary \ref{normality} 

		\begin{proof}
			Let $B = (b_{ij})_{n \times n} = (a_{ji})_{n \times n} = A^\dagger$. Clearly, $b_{i*} = a_{*i}^\dagger$ and $b_{*i} = a_{i*}^\dagger$ for all $i$. Note that,
			\begin{equation}
			\begin{split} 
			(AA^\dagger)_{ij} & = \sum_{k = 1}^n a_{ik}b_{kj} = \braket{a_{i*}| b_{*j}} = \braket{a_{i*}| a_{j*}} \\
			& = \sum_k w(v_{\mu i}, v_{\nu k}) w(v_{\nu k}, v_{\mu j}) : k \in \nbd(v_{\mu i}) \cap \nbd(v_{\mu j}).
			\end{split} 
			\end{equation} 
			Similarly, $(A^\dagger A)_{ij} = \sum_k w(v_{\nu i}, v_{\mu k}) w(v_{\mu k}, v_{\nu j}) : k \in \nbd(v_{\nu i}) \cap \nbd(v_{\nu j})$. Hence, we get the equality as stated for normality.
		\end{proof}
		
		Proof of theorem \ref{math_thm} 
		
		\begin{proof}
			The commutativity and normality conditions follow from the lemma \ref{commutativity} and \ref{normality} for all non-diagonal blocks. Note that, diagonal blocks are adjacency matrices of $\langle C_\mu \rangle$ which are Hermitian, hence normal. The degree condition includes all diagonal blocks in this family.
			
			First we consider commutativity of two diagonal blocks,
			\begin{equation}
			\begin{split} 
			& \frac{1}{d}(D_\mu \pm A_{\mu \mu}) \frac{1}{d}(D_\nu \pm A_{\nu \nu}) = \frac{1}{d}(D_\nu \pm A_{\nu \nu}) \frac{1}{d}(D_\mu \pm A_{\mu \mu})\\
			\Rightarrow & D_\mu D_\nu \pm D_\mu A_{\nu \nu} \pm A_{\mu \mu}D_\nu + A_{\mu \mu}A_{\nu \nu} = D_\nu D_\mu \pm D_\nu A_{\mu \mu} \pm A_{\nu \nu}D_\mu + A_{\nu \nu}A_{\mu \mu}\\
			\Rightarrow & (A_{\mu \mu}A_{\nu \nu} - A_{\nu \nu}A_{\mu \mu}) \pm (D_\mu A_{\nu \nu} - A_{\nu \nu}D_\mu) \pm (A_{\mu \mu}D_\nu - D_\nu A_{\mu \mu}) = 0\\
			\Rightarrow & (A_{\mu \mu}A_{\nu \nu} - A_{\nu \nu}A_{\mu \mu})_{ij} \pm (D_\mu A_{\nu \nu} - A_{\nu \nu}D_\mu)_{ij} \pm (A_{\mu \mu}D_\nu - D_\nu A_{\mu \mu})_{ij} = 0.
			\end{split}
			\end{equation}
			In terms of graphical parameters we may write,
			\begin{equation}
			\begin{split}
			(D_\mu A_{\nu \nu} - A_{\nu \nu}D_\mu)_{ij} & = d_{\mu i}(A_{\nu \nu})_{ij} - (A_{\nu \nu})_{ij} d_{\mu j} = w(v_{\nu i},v_{\nu j})(d_{\mu i} - d_{\mu j}),
			\end{split}
			\end{equation}
			\begin{equation}
			\begin{split}
			(A_{\mu \mu}D_\nu - D_\nu A_{\mu \mu})_{ij} & = (A_{\mu \mu})_{ij}d_{\nu j} - d_{\nu i}(A_{\nu \nu})_{ij} = w(v_{\mu i}, v_{\mu j})(d_{\nu j} - d_{\nu i}).
			\end{split} 
			\end{equation}
			Also from the corollary \ref{commutativity_2}, 
			\begin{equation}
			\begin{split} 
			(A_{\mu \mu}A_{\nu \nu} - A_{\nu \nu}A_{\mu \mu})_{ij} & = \sum_{k \in \nbd(v_{\mu i}) \cap \nbd(v_{\nu j})} w(v_{\mu i}, v_{\mu k}) w(v_{\nu k}, v_{\nu j}) \\
			& - \sum_{k \in \nbd(v_{\nu i}) \cap \nbd(v_{\mu j})} w(v_{\nu i}, v_{\nu k}) w(v_{\mu k}, v_{\mu j}).
			\end{split} 
			\end{equation}
			Thus for commutativity of diagonal blocks the following degree condition need to be satisfied,
			\begin{equation}
			\begin{split} 
			\sum_{k \in \nbd(v_{\mu i}) \cap \nbd(v_{\nu j})} & w(v_{\mu i}, v_{\mu k}) w(v_{\nu k}, v_{\nu j}) - \sum_{k \in \nbd(v_{\nu i}) \cap \nbd(v_{\mu j})} w(v_{\nu i}, v_{\nu k}) w(v_{\mu k}, v_{\mu j}) \\ 
			& \pm \big[w(v_{\nu i},v_{\nu j})(d_{\mu i} - d_{\mu j}) + w(v_{\mu i}, v_{\mu j})(d_{\nu j} - d_{\nu i}) \big] = 0.
			\end{split} 
			\end{equation}
			We consider $+$ for $\rho_q(G)$ and $-$ for $\rho_l(G)$ in the above equation.
			
			Now we consider commutativity of a diagonal and a non-diagonal block.
			\begin{equation}
			\begin{split} 
			& \frac{1}{d}(D_{\mu} \pm A_{\mu \mu})\frac{\pm 1}{d}A_{\alpha \beta} = \frac{\pm 1}{d}A_{\alpha \beta} \frac{1}{d}(D_{\mu} \pm A_{\mu \mu})\\
			\Rightarrow & D_{\mu}A_{\alpha \beta} \pm A_{\mu \mu}A_{\alpha \beta} = A_{\alpha \beta}D_{\mu} \pm A_{\alpha \beta}A_{\mu \mu}.
			\end{split} 
			\end{equation}
			Rearranging the terms we get the equation,
			\begin{equation}
			(D_{\mu}A_{\alpha \beta} - A_{\alpha \beta}D_{\mu}) \pm (A_{\mu \mu}A_{\alpha \beta} - A_{\alpha \beta}A_{\mu \mu}) =  0.
			\end{equation}
			The above equation holds if for all $i,j$ with $1 \le i,j \le n$,
			\begin{equation}
			\begin{split} 
			(D_{\mu}A_{\alpha \beta})_{ij} & - (A_{\alpha \beta}D_{\mu})_{ij} \pm \{ (A_{\mu \mu}A_{\alpha \beta})_{ij} - (A_{\alpha \beta}A_{\mu \mu})_{ij} \} =  0\\
			\Rightarrow d_{\mu i}(A_{\alpha \beta})_{ij} & - (A_{\alpha \beta})_{ij}d_{\mu j} \pm \{ (A_{\mu \mu}A_{\alpha \beta})_{ij} - (A_{\alpha \beta}A_{\mu \mu})_{ij} \} =  0.
			\end{split}
			\end{equation}
			Graph theoretic counterpart of $(A_{\mu \mu}A_{\alpha \beta} - A_{\alpha \beta}A_{\mu \mu})$ follows from the corollary \ref{commutativity_1}. Thus,
			\begin{equation}
			\begin{split} 
			(A_{\mu \mu}A_{\alpha \beta})_{ij} - (A_{\alpha \beta}A_{\mu \mu})_{ij} & = \sum_{k \in \nbd(v_{\mu i}) \cap \nbd(v_{\beta j})} w(v_{\mu i}, v_{\mu k}) w(v_{\alpha k}, v_{\beta j})\\ 
			& - \sum_{k \in \nbd(v_{\alpha i}) \cap \nbd(v_{\mu j})} w(v_{\alpha i}, v_{\beta k}) w(v_{\mu k}, v_{\mu j}).
			\end{split} 
			\end{equation} 
			Also,
			\begin{equation}
			d_{\mu i}(A_{\alpha \beta})_{ij} - (A_{\alpha \beta})_{ij}d_{\mu j} = w(v_{\alpha i}, v_{\beta j}) (d_{\mu i} - d_{\mu j}).
			\end{equation}
			Combining the above two equations we get,
			\begin{equation}
			\begin{split} 
			w(v_{\alpha i}, v_{\beta j}) (d_{\mu i} - d_{\mu j}) & \pm \big[\sum_{k \in \nbd(v_{\mu i}) \cap \nbd(v_{\beta j})} w(v_{\mu i}, v_{\mu k}) w(v_{\alpha k}, v_{\beta j}) \\
			& - \sum_{k \in \nbd(v_{\alpha i}) \cap \nbd(v_{\mu j})} w(v_{\alpha i}, v_{\beta k}) w(v_{\mu k}, v_{\mu j}) \big] = 0.
			\end{split} 
			\end{equation}
			
			The density matrix $\rho(G)$ represents a zero discord quantum state if its blocks $B_{\mu \nu}$ form a family of commuting normal matrices \cite{huang2011new}. In the graph theoretic context, $\rho(G)$ meets the criterion provided the above three conditions are satisfied.
		\end{proof}

	\section*{Acknowledgement}
		This work was partially supported by the project \textit{“Graph theoretical aspects in quantum information processing”} [Grant No. 25(0210)/13/EMR-II] funded by Council of Scientific and Industrial Research, New Delhi. S.D. is grateful to the Ministry of Human Resource Development, Government of India, for a doctoral fellowship. This work may be a part of his doctoral thesis.


\end{document}